\documentclass[12pt]{article}
\usepackage[backend=bibtex, style=numeric]{biblatex}
\usepackage{amsmath, amsthm, amssymb,amscd, bbm}
\usepackage{fullpage, paralist,enumerate}
\usepackage{subfig}
\usepackage{verbatim}
\usepackage[all]{xy}
\usepackage[parfill]{parskip}
\usepackage{graphicx}
\usepackage{mathabx}
\CompileMatrices

\usepackage[shortlabels] {enumitem}
\usepackage{xcolor}
\usepackage{hyperref}
\hypersetup{
    colorlinks   = true,
    citecolor    = blue,
    urlcolor=blue,
}

\newcommand{\R}{{\mathbb R}}

\newcommand{\F}{{\mathcal F}}
\newcommand{\cP}{{\mathcal P}}
\newcommand{\G}{{\mathcal G}}
\newcommand{\E}{{\mathbb E}}

\newcommand{\A}{{\mathcal A}}

\newcommand{\cF}{{\mathcal F}}

\newcommand{\cG}{{\mathcal G}}

\newtheorem{theorem}{Theorem}[section]
\newtheorem{lemma}[theorem]{Lemma}
\newtheorem{proposition}[theorem]{Proposition}
\newtheorem{corollary}[theorem]{Corollary}

\theoremstyle{remark}

\newtheorem{definition}[theorem]{Definition}

\newtheorem{assumption}[theorem]{Assumption}

\theoremstyle{remark}

\newcommand{\conv}{{\mathrm{conv}\,}}
\newcommand{\cl}{{\mathrm{cl}\,}}
\newcommand{\dom}{{\mathrm{dom}\,}}
\newcommand{\setint}{{\mathrm{int}\,}}

\usepackage{authblk}

\title{Multistage Portfolio Optimization: A Duality Result
in Conic Market Models} 
\author[1]{Robert Bassett}
\author[2]{Khoa L\^e}
\affil[1]{University of California, Davis}
\affil[2]{University of Calgary}
\begin{document} 
\maketitle
\begin{abstract}
We prove a general duality result for multi-stage portfolio optimization
problems in markets with proportional transaction costs. The financial
market is described by Kabanov's model of foreign exchange markets
over a finite probability space and finite-horizon discrete
time steps. This framework allows us to compare vector-valued 
portfolios under a partial ordering, so that our model does not
require liquidation into some numeraire at terminal time.

We embed the vector-valued portfolio problem into the set-optimization
framework, and generate a problem dual to portfolio optimization.
Using recent results in the development of set optimization, we then show that a
strong duality relationship holds between the problems.
\end{abstract}
\section{Introduction}

Portfolio optimization problems have a long and rich history.
Traditionally, portfolio optimization took place in models without
transaction costs, as in \cite{Merton}. Portfolio optimization
problems in financial market models which include proportional
transaction costs first appeared in the work of Magill and Constantinides 
\cite{Magill}. The two asset
model was solved rigorously and improved upon shortly thereafter
\cite{Davis} \cite{SonerShreve}. Since then, much more focus has
been placed on deriving results in markets with transaction costs, as
researchers tried to develop results analogous to the classical case.

In this paper, we consider the conic market model initially developed
by Kabanov \cite{K99} for modeling foreign transaction market. The
conic market model 
expresses portfolios in terms of the number of physical
assets they contain, as opposed to their values. This allows the
formulation of wealth processes without the explicit use of
stochastic integration. Though this quality may be seem surprising to
portfolio optimization veterans, it is attractive in terms of its simplicity
and intuitive nature. In order to compare portfolios, we consider them
as assets of vectors under a partial ordering, and invoke the theory
of set optimization to formulate the portfolio problem.


Set optimization is primarily motivated by the desire to optimize with
respect to a non-total order relation. Vector-valued optimization fits
directly into this framework, with component-wise comparison. However,
we prefer to work in a set-optimization framework as opposed to the
less general vector-optimization because of its succinct theory
\cite{SetOpt}. Following the work of 
\cite{Ham09}, we introduce an ordering of sets which
generates a complete lattice. This allows us to define corresponding
notion of infimums and supremums of sets--a fundamental step in the
formulation of the portfolio optimization problem. We then use the
tools in \cite{HaLo14} to formulate a set-valued dual to the
portfolio optimization problem. Because we consider the multistage
problem, our results are generalize those in
\cite{Sophie}.

Constructing a primal-dual pair of problems for set-valued
portfolio optimization provides insight into the relationship
between traditional portfolio optimization theory and the
proportional transaction cost case. But it is also our hope that the
results contained here would be of more than just theoretical
interest. Recent work \cite{primalDual} \cite{hamel2014benson} has been
investigating computational techniques for solving set-valued
optimization problems. In particular, \cite{primalDual} uses both the
primal and dual formulation of a set optimization problem to work
towards computing a solution. In this sense, the results in this paper
provide a valuable relationship which will help bring the portfolio
optimization problem considered closer practioners, as computational
techniques progress.

The rest of this paper is organized as follows. In the first section,
we introduce the material which we use to formulate the set-valued
portfolio optimization problem. This includes a review of the conic market
model in addition to a summary of the set-optimization tools and
techniques the we will use in our problem formulation. The next
section explicitly formulates the multi-period utility maximization
problem. In the third section, we discuss duality in the set
optimization framework. The fourth section is devoted to our main
results, the formulation of a dual problem and a proof that a strong
duality relationship holds. The last section applies the main results
to an example utility maximization problem.

\textbf{Acknowledgment:} The paper was initiated while the authors attended the workshop Mathematics Research Communities 2015 at Snowbird. The authors thank Birgit Rudloff and Zach Feinstein for their encouragement and guidance into the subject. Part of the paper was completed when the second author was in residence at the Mathematical Sciences Research Institute in Berkeley, during which he was supported by the National Science Foundation under Grant No. DMS-1440140. 

\section{Preliminaries}

\subsection{Conical Market Model}
In this section, we recall the framework of the conic market model
with transaction costs introduced in \cite{K99}, though we primarily 
follow the development in \cite{Sch04}.

Consider a financial market which consists of $d$ traded assets. In
classical models, we assume that at some terminal time $T$ all assets
are liquidated, i.e. converted to some numeraire. In certain
applications, this is unrealistic. For example, an agent with a 
portfolio consisting of
assets in both US and European markets should not need to choose
between liquidation into Euro or USD to establish its relative value.
For this reason, we use a numeraire-free approach by considering
vector-valued portfolios. In particular, we express portfolios in
terms of the number of physical units of each asset, instead of the
value of those assets with respect to a numeraire. This approach is
especially interesting when liquidation into some numeraire has
an associated transaction cost. In this case conversion to a unified 
currency is irreversible, and different choices of numeraire could
result in different relative values of portfolios.

We consider a market in which transaction costs are proportional to
the number of units exhanged. To model these costs, we introduce
the notion of a bid-ask matrix.
\begin{definition}
A \emph{bid-ask matrix} is a $d \times d$ matrix $\Pi$ such that its
entries $\pi^{ij}$ satisfy
\begin{enumerate}
\item $\pi^{ij} > 0$, for $1 \leq i, j \leq d$
\item $\pi^{ii}=1$, for $1 \leq i \leq d$.
\item $\pi^{ij} \leq \pi^{ik} \pi^{kj}$, for $i \leq i, j, k \leq d$.
\end{enumerate}
\end{definition}
The terms of trade in the market are given by the bid-ask matrix, in
the sense that the entry $\pi^{ij}$ gives the number of units of asset
$i$ which can be exchanged to one unit of asset $j$. Thus the pair
$\{\frac{1}{\pi^{ji}}, \pi^{ij} \}$ denotes the bid and ask prices of
the asset $j$ in terms of the asset $i$. A financial interpretation of
the first and second properties of a bid-ask matrix is straightforward, with the third
condition ensuring that an agent cannot achieve a better exchange rate
through a series of exchanges than exchanging directly.

Next, we consider the notions of solvency and available portfolios.
Recall that, given a set $C \subseteq \R^d$, the \emph{convex cone
generated by $C$} is the set
$$\text{cone}(C)=\{ \sum_{i = 1}^{n} \lambda_{i} y_{i} : y_{i} \in C, \lambda_{i}
\geq 0, 1 \leq i \leq n, n \in \mathbb{N} \}.$$

\begin{definition} For a given bid-ask matrix $\Pi$, the
\emph{solvency cone} $K(\Pi)$ is the convex cone in $\R^{d}$ generated by
the unit vectors $e^i$ and $\pi^{ij} e^{i} - e^{j}$, $1 \leq i, j \leq d.$
\end{definition}
Solvent positions in vector valued portfolios are those which can be
traded to the zero portfolio.
The vector $\pi^{ij} e^{i} - e^{j}$, which consists of $\pi^{ij}$ 
long position in asset $i$ and one short in asset $j$, is solvent
because the terms of trade given by $\Pi$ allow exchanging $\pi^{ij}
e^{i}$ to $e^{j}$. It follows that any non-negative linear
combination of $\pi^{ij} e^{i} - e^{j}$ is also solvent. We also allow
an agent to discard non-negative quantities of an asset in order to
trade to the zero portfolio, which justifies including the $e^i$
vectors in the solvency cone definition.

What is that set of portfolios that can be obtained from the zero
portfolio, according
to the terms of trade governed by $\Pi$? Similar to the definition
of the solvency cone, it consists of vectors $e^{j} - \pi^{ij}
e^{i}$, which correspond to trades at the exchange rate given by
$\Pi$. Again permitting trades where agents discard resources, we see
that the set of
portfolios available at price zero is the cone $-K(\Pi)$.

Given a cone $K \subseteq \R^d$, we denote by $K^{+}$ the positive
polar cone of $K$, i.e.,
$$K^{+} = \{w \in \mathbb{R}^{d} : \langle v, w \rangle \geq 0 \text{
for } v \in K(\Pi) \}.$$
Recall that the interior of $K^{*}$ is the set 
$$\setint K^{*} = \{w : \langle w,x\rangle >0, \forall x \in K, x \neq 0 \}.$$

\begin{definition}
A nonzero element $w \in \mathbb{R}^d$ is a \emph{consistent price system}
for the bid-ask matrix $\Pi$ if $w$ is in the positive polar cone of
$K(\Pi)$, so that
$$ w \in K^{+}(\Pi) = \{w \in \mathbb{R}^{d}: \langle v, w \rangle
\geq 0 \text{ for } v \in K(\Pi) \}.$$
The set of all consistent price systems for a bid-ask matrix is then
simply $K^{+}(\Pi) \setminus \{ 0 \}$.
\end{definition}

The notion of a consistent price system has an important financial
interpretation. A price system $w$ gives a non-negative price
 $w^{i}$ for each asset $i$. One interpretation of the definition of a
consistent pricing system is
that, if we fix some numeraire asset $i$, then $w$ satisfies the 
condition that the frictionless exchange rate $\frac{w^{j}}{w^{i}}$
for asset $j$ is less than $\pi^{ij}$. Allowing for arbitrary choice
of numeraire $i$ and asset $j$, this is equivalent to
$$\{w \in \mathbb{R}^{d}_{+} \setminus \{ 0 \}: \frac{w^{j}}{w^{j}} \geq \pi^{ij} \text{ for } 1 \leq i, j \leq d \}.$$
One can easily show that this set is in fact equal to
$K^{+}(\Pi) \setminus \{ 0 \}$, the set of all price systems
consistent with $\Pi$.

Fixing a filtered probability space $(\Omega, (\cF)_{t=0}^{T},
\mathbb{P})$, we model a financial market by $(\Pi_{t})_{t=0}^{T}$, an
$\cF$ adapted process taking values in the set of bid-ask matrices. Such 
a process will be called a bid-ask process. We make the following
simplifying assumptions.
\begin{assumption}\label{assumptions} $(\Omega,(\cF)_{t=0}^{T}, \mathbb{P})$ satisfies
		\begin{itemize} 
		\item $\mathcal{F}_{0}= \{\emptyset, \Omega\}$ is trivial.
		\item The model is in discrete time with $t=0,...,T$
		\item The probability space $\Omega$ is finite, with $|\Omega|=N$ 
		\item Each element in $\Omega$ has nonzero probability, i.e.
		$\mathbb{P}[\omega_{n}]=p_n >0$, where $\Omega = \{ \omega_{1},
		\omega_{2},..., \omega_{N} \}$ and $n=1,..,N$.
		\end{itemize}
\end{assumption}
The assumption that $\Omega$ is finite means that
 we can identify the space of all
$d$-dimensional random variables with $\R^{d \times N}$ and inner
product $\mathbb{E} \langle x, y \rangle$, where $x$ and $y$ are random vectors.
As a result, the different
topologies $L^{\infty}(\Omega, \mathcal{F}, \mathbb{P})$, $L^{1}(\Omega,
\mathcal{F}, \mathbb{P})$, $L^{0}(\Omega, \mathcal{F}, \mathbb{P})$,
etc. on
the set of all $\mathbb{R}^d$-valued random variables $X: \Omega \to
\mathbb{R}^d$
are isomorphic. We will refer to this topology simply as
$L^{p}(\Omega, \mathcal{F}, \mathbb{P})$ for some $p \in (1, \infty)$,
in order to make clear when we are referring to the dual space
$L^{q}(\Omega, \mathcal{F}, \mathbb{P})$ where $\frac{1}{p} +
\frac{1}{q} =1$. For ease of notation, we will denote these
spaces $L^{p}(\mathcal{F}; \mathbb{R}^d)$ and $L^{q}(\mathcal{F};
\mathbb{R}^d)$. For the sake of notation, we will denote the components of a vector
$x \in \R^{d \times N}$ by $x_{i}(\omega)$ for $\omega \in \Omega$, $1
\leq i \leq N.$

Again exploiting the finiteness of $\Omega$, we know that any cone
in $L^{p}(\cF_t, \R^d)$ generated by a finite set of random vectors $\{x_{i}
\}_{i=1}^{m}$ is generated by $\{x_{i} \mathbbm{1}_{\Gamma_{j}}\}$ in
$\R^{d \times N}$, where $\{\Gamma_{j}\}_{j \in J}$ is the set of atoms of $\cF_{t}$.

Let $(\Pi_{t})_{t=0}^{T}$ be a bid-ask process. This generates a
cone-valued process $(K_{t})_{t=0}^{T}$ where each $K_{t}$ is an
associated solvency cone. We denote by $L^{p}(\mathcal{F}_{t}, K_{t})$
the set
$$ \{X \in L^{p}(\mathcal{F}_{t}; \mathbb{R}^d): X_{t}(\omega) \in
K_{t}(\omega) \}.$$
for each $\omega \in \Omega$.

We can now define the notion of a self-financing portfolio.
\begin{definition}
An $\R^d$-valued adapted process
$\vartheta=(\vartheta_t)_{t=0}^T$ is called a self-financing portfolio
process if the increments
 \begin{equation*}
\xi_t(\omega):=\vartheta_t(\omega)- \vartheta_{t-1}(\omega)
\end{equation*}
belong to the cone $-K_t(\omega)$ of portfolios
available at price zero, for all time $t=0,\dots,T$. We also put
$\vartheta_{-1}=0$ by convention.
\end{definition}
For each $t=0,\dots, T$, we denote by $A_t$ the convex cone in
$L^p(\cF_t;\R^d)$ formed by the random variables $\vartheta_t$, where
$\vartheta=(\vartheta_i)^T_{i=0}$ runs through the self-financing
portfolio processes. We always assume that the initial portfolio
$\vartheta_0$ is deterministic. $A_T$ may be then interpreted as the
set of positions available at time $T$ from an initial endowment
$0\in\R^d$. More precisely, if we denote by $\mathbbm{1} \in
L^{p}(\mathcal{F}_{0}, \mathbb{R}^d)$ the constant random variable
that assumes the value $1$, we have the following result:
\begin{proposition}
For each $t=1,...,T$,
\begin{equation*}
A_t=-K_0 \mathbbm{1}-L^p(\cF_1;K_1)-\cdots-L^p(\cF_t;K_t)\,.  
\end{equation*} 
\end{proposition}
\begin{proof}
Assume $x \in A_t$. Then $x$ is a convex combination of random
variables 
$$x=\sum_{j=1}^{n} \lambda^{j} \vartheta^{j}_{t}$$ 
where for each $j$, $(\vartheta_{i}^{j})_{i=0}^{T}$ is a self-financing 
portfolio and $\lambda^{j} \geq 0$. We rewrite $x$ as
$$x = \sum_{j=1}^{n} \lambda^{j} \left( \vartheta^{j}_{t} -
\vartheta^{j}_{t-1} + \vartheta^{j}_{t-1} - ... -\vartheta^{j}_{0} +
\vartheta^{j}_{0} \right).$$
Expanding this sum, each $\sum_{j=1}^{n} \lambda^{j}(\vartheta^{j}_{i}
-\vartheta^{j}_{i-i}) \in -K(\Pi_{i})$ since $-K(\Pi_{i})$ is a cone and
$\vartheta^{j}_{i} - \vartheta^{j}_{i-1} \in -K(\Pi_{i})$ for every
$j$. We have established that
$$A_t \subseteq -K_0 \mathbbm{1}-L^p(\cF_1;K_1)-\cdots-L^p(\cF_t;K_t)\,.$$  
The reverse containment follows by a symmetric argument and is
omitted.
\end{proof}

Similarly, if one starts with an initial endowment $x_0\in\R^d$,
then the collection of all random portfolios available at time $T$ is
given by $A_T(x_0)=x_0 \mathbbm{1}+A_T$. Explicitly, we have
\begin{equation}\label{def:ATx}
A_T(x_0)=x_0  \mathbbm{1} -K_0 \mathbbm{1}-L^p(\cF_1;K_1)-\cdots-L^p(\cF_T;K_T)\,,
\end{equation} with convention $A_T(0)=A_T$.

Another important concept in a financial market model is the concept
of arbitrage.  In the conic market model framework, the bid-ask
process
$(\Pi_t)_{t=0}^T $ is said to satisfy the no arbitrage property
if
\begin{equation}
\label{con:na} 
A_T\bigcap L^p(\cF_T;\R^d_+)=\{0\}\,.
\end{equation}
We will assume that our market model satisfies the no arbitrage property.

In classical financial market models, no arbitrage is intimately
connected to the existence of an equivalent martingale measure. The
corresponding notion in the conic market model is a consistent pricing
process.
\begin{definition} An
adapted $\R^d_+$-valued process $(Z_t)_{t=0}^T$ is called a consistent
 pricing process for the bid-ask process
$(\Pi_t)_{t=0}^T$ if $Z$ is martingale and $Z_t(\omega)$ lies in
$K_t(\omega)^{+} \setminus\{0\}$ for each $t=0,\dots,T$.
\end{definition}

The following extension of the Fundamental Theorem on Asset Pricing, due to
Kabanov and Stricker, \cite{KS01}, establishes the connection between no arbitrage and
consistent pricing processes.
\begin{theorem} \label{theor:FTAP}
Let $\Omega$ be finite. The bid-ask process $(\Pi_{t})_{t=0}^{T}$ 
satisfies the no arbitrage condition if and only if there is a
consistent pricing system $Z=(Z_{t})^{T}_{t=0}$ for
$(\Pi_{t})_{t=0}^{T}$.
\end{theorem}
This theorem is a fundamental component in the proof of our main
duality result, Theorem \ref{thm:DualProb}.

We will make use of this result in the form of the following lemma 
\begin{lemma} \label{lem:KhoaLemma} 
A bid-ask process satisfies the no arbitrage condition if and only if
$$(-A_{T}^{+}) \bigcap \setint(L^{q}(\cF_{T}, \R^{d}_{+})) \neq
\emptyset.$$
\end{lemma}
\begin{proof}
For the forward direction, assume a bid-ask process satisfies the no
arbitrage condition. Then 
$$A_{T} \cap L^{p}(\cF_{T}, \R_{+}^{d}) = \{ 0 \}.$$
Since $\Omega$ is finite, $A_{T}$ is the sum of finitely
generated closed convex cones, so it is a finitely generated convex
cone, and hence closed. Let
$C:=\conv(\{e(\omega)_{i}, 1 \leq i \leq d, \omega \in \Omega \})$,
where each $\{e(\omega)_{i}$ is a unit vector in $\R^{d \times N}$. Since
$C$ is the convex hull of a finite set of points in $L^{p}(\cF_{T}, \R_{+}^{d})$, it is
compact. Obviously $0 \not \in C$. By the separation theorem (in the
case that one set is closed and the other compact) there is a
nonzero $z \in \R^n$ that strictly separates $C$ and $A_T$. That is,
$$\sup_{x \in A_T} \langle z , x \rangle < \inf_{y \in C} \langle z, y
\rangle.$$
$C$ is compact, so the expression on the right-hand side of the
inequality is finite. Since $A_T$ is a cone, the left-hand side of
the inequality must then be zero, so $z \in -A_T^{+}$. Furthermore, 
$\langle z, y \rangle >0$ for each $y \in C$, so that $\langle z,
\lambda y \rangle > 0$ for all $y \in C$ and $\lambda \neq 0$. Since
$C$ generates $L^{p}(\cF_T)$, we have that $\langle z, \lambda y
\rangle >0$ for all $y \in L^{p}(\cF_{T}, \R_{+}^{d})$ with $y \neq
0$. It follows that 
$$z \in \setint \left( (L^{p}(\cF_{T},
\R_{+}^{d}))^{+} \right) = \setint \left( L^{q}(\cF_{T}, \R_{+}^{d})
\right).$$

For the reverse direction, assume that 
$$(-A_{T}^{*}) \cap \setint(L^{q}(\cF_{T}, \R^{d}_{+}))
\neq\emptyset.$$
Then there is a $z$ such that 
$$\langle z, x \rangle \leq 0 \text{ for }x \in A_{T} \text{ and } \langle
z, y \rangle > 0 \text{ for } y \in K \setminus \{ 0 \}.$$
This obviously implies that $A_T$ and $K \setminus \{0 \}$ are
disjoint.
\end{proof}
\subsection{Set Optimization}

In this section, we review the components of set-valued optimization
that will be necessary to introduce the portfolio optimization problem.
For a more detailed exposition of
the set-valued optimization framework and the corresponding duality
theory, see \cite{SetOpt, Ham09, HaLo14}.

We begin by constructing a suitable notion of ``order'' for sets.
Let $Z$ be a non-trivial real linear space. Given a convex cone $C 
\subsetneq Z$ with $0 \in C$, we have a preordering of
$Z$, denoted by $\leq_{C}$, which is defined as
$$z_{1} \leq _{C} z_{2} \iff z_{2} - z_{1} \in C$$
for any $z_{1}, z_{2} \in Z$. The following are equivalent to $z_{1}
\leq_{C} z_{2}$,
$$z_{1} \leq_{C} z_{2} \iff z_2 - z_1 \in C \iff z_2 \in z_1 +C \iff
z_1 \in z_2 -C.$$
These last two expressions can be used to extend $\leq_{C}$ from $Z$
to $\cP(Z)$, the power set of $Z$. Given $A,B \in \cP(Z)$, we define two
possible extensions.
$$A \preccurlyeq_{C} B \iff B \subseteq A+C$$
$$A \curlyeqprec_{C} B \iff A \subseteq B-C.$$
We use $+$ to denote Minkowski addition of sets, with set convention
that $A+ \emptyset = \emptyset+A= \emptyset$ for all $A \in \cP(Z)$.

In what follows we will exclusively discuss the relation
$\preccurlyeq_{C}$, which is appropriate for set-valued minimization. Each of the
results we include has a corresponding result for $\curlyeqprec_{C}$ in the
maximization context, which we omit. For further details see
\cite{SetOpt}. 

In addition, we assume that $Z$ is equipped with a Hausdorff, locally
convex topology. We consider the space
$$\G(Z,C) = \{A \subseteq Z: A = \cl \conv(A+C) \}$$
where $\cl \conv$ is the closure of the convex hull. We abbreviate
$\G(Z, C)$ to $\G(C)$, when $Z$ is clear from the context. We define an
associative and commutive binary operation $\oplus: \G(C) \times \G(C)
\to \G(C)$ by 
$$A \oplus B = \cl (A+B)$$

Observe that in
$\G(C)$, the relation $\preccurlyeq_{C}$ reduces to containment. For any $A, B
\in \G(C)$, 
$$A \preccurlyeq_{C} B \iff A \subseteq B.$$
As shown in \cite{Ham09}, the pair $(\G(C), \supseteq)$ is a complete
lattice, meaning that $\supseteq$ yields a partial order on $\G(C)$,
and that each subset of $\G(C)$ has an infimum and supremmum with respect to
$\supseteq$ in $\G(C)$. Given $\emptyset \neq \A \subseteq \G(C)$, the
infimal and supremal elements in $\G(C)$ are
\begin{equation}
\label{def:infsup}
\inf_{(\G(C), \supseteq)} \A = \cl \conv \bigcup_{A \in \A} A, \; \; \;
\sup_{(\G(C), \supseteq)} \A = \bigcap_{A \in \A} A.
\end{equation}

In order to preserve intuition, 
it is useful to recall how this framework relates to the familiar
complete lattice of the extended real numbers $\R \cup \{ \pm \infty \}$ with
the $\leq$ order. The extended real-numbers translate into the
set-valued framework described above by using the ordering cone $C =
\R_{+}$ and identifying each point $z \in \R$ with the set $\{z\}+
\R_{+}$ in $(\G(\R, \R_{+}), \supseteq)$. Moreover, $+ \infty$ and $-
\infty$ in the usual framework are replaced by $\emptyset$ and $\R$,
repectively, in the set-valued case.

Next, assume that $X, Y$ are two locally convex spaces, and that $D
\subseteq Y$ is a convex cone with $0 \in D$. Let $f:X \to \G(C)$ and
$g:X \to \G(D)$ be two set-valued functions. We consider optimization problems
of the form
$$\min_{x \in X} f(x) \text{ subject to } 0 \in g(x).$$
Where the minimum refers to the set-valued ordering previously discussed. In
other words, we want to find the set
$$\inf_{(\G(C), \supseteq)} \{f(x) | x \in X, 0 \in g(x) \} = \cl \conv
\bigcup \{ f(x) | x \in X, 0 \in g(x) \}.$$
This is the minimum of our optimization problem.

Extending the notion of a minimizer to the set-valued case is slightly
more subtle. Given $f: X \to \G(C)$ and $M \subseteq X$,  we denote the
set of all values of $f$ on $M$ by
$$f[M] = \{ f(x) | x \in M \}.$$
The minimal elements of $f[M]$ are defined by
$$\text{Min } f[M] := \{f(x) | f(x) \in f[M] \text{ and } \forall f(y) \in f[M] \text{ with } f(y)
\supseteq f(x), f(y) = f(x) \}.$$
Similarly, an element $\overline{x}$ is a minimizer of $f$ on $M$ if
$f(\overline{x}) \in \text{Min } f[M]$.

In addition to a minimality condition, we also expect a solution to
attain the infimum of a problem. We say that the infimum of a problem
$$\min f(x) \text{ subject to } x \in X$$
is \emph{attained} at a set $\overline{X} \subseteq X$ if
$$\inf_{x \in \overline{X}} f(x) = \inf_{x \in X} f(X).$$
As per the definition of infimimum in \eqref{def:infsup}, this means
that
$$\cl \conv \bigcup_{x \in \overline{X}} f(x) = \cl \conv \bigcup_{x
\in X} f(x).$$
Alternatively, we say that the set $\overline{X}$ is an
\emph{infimizer} of the problem. 
Combining both of these requirements, we arrive at an appropriate
notion of a solution to a set optimization problem. 

\begin{definition}
Given $f: X \to \G(C)$, an infimizer $\overline{X} \subseteq X$ is
called a  \emph{solution} to the problem 
$$\min f(x) \text{ subject to } x \in X$$
if $\overline{X} \subseteq \text{Min } f[X].$
Similarly, we call an infimizer $\overline{X} \subseteq X$ a
\emph{full solution} to the problem if $\overline{X} =
\text{Min } f[X]$.
\end{definition}
In the typical optimization framework of the extended real numbers,
the notion of an infimizer and minimizer coincide because the search 
for infimizers can be reduced to singleton sets. In the
set-optimization setting, this is not the case, which warrants the above
definition. The infimum of a problem is, in general, the closure of
the union of function values, which is not necessarily a function
value itself. Further details and a more in depth review of this
issue can be found in \cite{HaLo14}.

We next review some important convex-analytic type properties for
set-valued functions.
\begin{definition} \label{def:Convex}
A set valued function $f: X \to (\G(C), \supseteq)$ is said to be
\emph{convex} if for every pair $x_1$, $x_2 \in X$ and every $t \in
(0,1)$
$$f(t x_1 + (1-t) x_2) \supseteq t f(x_1) + (1-t) f(x_2).$$
\end{definition}

It is straight-forward to show that convexity of $f$ is equivalent to
convexity of the graph of $f$, where
$$\text{graph } f := \{(x, z) \in X \times Z | z \in f(x) \} \subseteq
X \times Z.$$


We end this section with the following results, found in
\cite{HaLo14}, which use convexity to
simplify the computation of infimums and Minkowski sums.

\begin{proposition}
If $f:X \to \cP(Z)$ is convex and 
$$f(x) = \cl(f(x) + C),$$
then $f(x) \in \G(C)$
\end{proposition}
\begin{proof}
We want to show that for each $x \in X$,
 $f(x) = \cl \conv (f(x) +C)$, given that $f$ is convex and $f(x) =
\cl (f(x) +C)$. It suffices to show that $f(x)+C$ is convex, which
will follow if $f(x)$ is convex because the Minkowksi sum of two
convex sets is convex \cite{ConvAnal}. But $f(x)$ is convex for every
$x$, since for arbitrary $z_1, z_2 \in f(x)$, $t \in (0,1)$, we have
$$t z_1 + (1-t) z_2 \in t f(x) + (1-t) f(x) \subseteq f(tx + (1-t) x)
= f(x)$$
where the last containment comes from the convexity of $f$.
\end{proof} 

\begin{proposition}
If $f:X \to \G(C)$ and $g:X \to \G(D)$ are convex, then
$$\inf_{(\G(C), \subseteq)} \{ f(x) |x \in X, 0 \in g(x) \} = \cl
\bigcup \{ f(x) | x \in X, 0 \in g(x) \},$$
so that the convex hull can be removed from the definition of infimum.
\end{proposition}
\begin{proof}
We want to show that 
$$\bigcup \{ f(x) | x \in X, 0 \in g(x) \}$$
is convex. We begin by showing that $\{x \in X | 0 \in g(x) \}$ is
convex. If $0$ is contained in both $g(x_1)$ and $g(x_2)$, then for
any $t \in (0,1)$,
$$0 \in t g(x_1) + (1-t) g(x_2) \subseteq g(t x_1 + (1-t) x_2),$$
so that $\{x \in X | 0 \in g(x) \}$ is convex.

 Next, assume that $z_1, z_2 \in \bigcup \{ f(x) | x \in X, 0 \in g(x)
\}$. Then there are $x_1, x_2$ such that $z_1 \in f(x_1)$ and $z_2 \in
f(x_2)$, with $0 \in g(x_1) \cap g(x_2)$. Thus for any $t \in (0,1)$
$$t z_1 + (1-t) z_2 \in t f(x_1) + (1-t) f(x_2) \subseteq f(t x_1 +
(1-t) x_2).$$
Our initial claim gives that $0 \in g(t x_1 + (1-t) x_2)$, so that  
$\bigcup \{ f(x) | x \in X, 0 \in g(x)$ is convex and we have our
result.
\end{proof}



\section{Problem Formulation}

In this section, we explicitly formulate the multi-period utility
maximization problem.

We consider a function $U(x): \mathbb{R}^{d} \to \mathbb{R}^{d}$ which
models the utility of an agent's assets $x$ at the terminal time $T$.
We make the following assumptions on $U$.
\begin{enumerate}
\item $U$ is a vector valued component-wise function
$$U(x) = (u_1(x_1), u_2(x_2), ..., u_d(x_d)), \; \; x = (x_1, x_2,
..., x_d) \in \mathbb{R}^d$$
where each $u_i: \mathbb{R} \to \mathbb{R}$. Note each $u_i$ is
real-valued, as opposed to extended real valued. Thus $U$ is 
defined even in the case of negative wealth.
\item Each $u_i$ is strictly concave, stricly increasing, and
differentiable. 
\item Marginal utility tends to zero when wealth tends to infinity, so
that
$$\lim_{x_i \to \infty} u'_{i}(x_i) = 0.$$
\item $u_i$ satisfies the \emph{Inada condition}, so that the marginal
utility tends to infinity when $x_i$ tends to the infimum of the
 domain of $u_i$. In other words,
$$\lim_{x_i \to - \infty} u_{i}'(x) = \infty.$$
\end{enumerate}

These assumptions are standard in the context of utility maximization
problems \cite{MathArb}.

Let the ordering cone $C = \mathbb{R}_{+}^d$. We define the objective
function $F: L^{p}(\cF_{T}, \R^d) \to \G(C)$ to be the expected utility 
of a random portfolio at terminal time.
$$F(x) =  \mathbb{E}[-U(x)] + \mathbb{R}_{+}^{d}.$$
The expectation is taken with respect to the probability space
$(\Omega, \cF_{T}, \mathbb{P})$.

Note that in the definition of $F$, we have recast the utility
maximization problem into a minimization framework. This is to
establish more consistency with the set-valued optimization 
tools developed in \cite{SetOpt}, \cite{Ham09}, and \cite{HaLo14},
which cast their results in the traditional minimization framework of
convex analysis. Of course, one could consider the maximization form
of the problem without any loss of generality.

The portfolio optimization problem then takes the form
$$\text{minimize }  F(x)$$
subject to the constraint that the portfolio $x$ is the terminal
result of a self-financing portfolio with initial endowment $x_0$. In
other words, we have the problem
\begin{align*}
\label{prob:P} \tag{$\mathcal{P}$}
\text{minimize }& F(x) \\
\text{subject to }& x \in A_{T}(x_0)
\end{align*}

\section{Duality in Set Optimization}

In this section we recall the necessary results from set-valued duality \cite{HaLo14}
which we will use to prove our main result.
\eqref{prob:P}. 

Set-valued Lagrange duality follows a similar theme to the real-valued
case. Given convex cones $C \subseteq Z$ and $D \subseteq Y$, and convex functions 
$f:X \to \cG(C) \subseteq \mathcal{P}(Z)$ and 
$g: X \to \cG(D) \subseteq \mathcal{P}(Y)$, we are interested in the primal problem
\begin{align*}
\label{prob:P_e} \tag{$\mathcal{P}_{e}$}
\text{minimize } f(x) \text{ subject to } 0 \in g(x).
\end{align*}
i.e. we search for a set $\bar p \subseteq Z$ where
$$\overline{p} := \inf_{(\cG(C), \supseteq)} \{f(x) | x \in X, 0 \in g(x) \} = \cl
\conv \bigcup \{f(x) | x \in X, 0 \in g(x) \}.$$
The first step is to define a set-valued Lagrangian function which recovers the
 objective, in the sense that the function $f(x)$ is the supremum of
the Lagrangian over the set of dual variables. 

For $y^{*} \in Y^{*}$ and $z^{*} \in Z^{*}$, where $Y^{*}$ and $Z^{*}$
denote the topological dual spaces of $Y$ and $Z$, respectively, define the set-valued
function $S_{(Y^{*}, Z^{*})}: Y \to \mathcal{P}(Z)$ by
$$S_{(y^{*}, z^{*})}(y) = \{ z \in Z | \; y^{*}(y) \leq z^{*}(z)\}
.$$
We use these functions to formulate a Lagrangian function.
\begin{definition} \label{def:lag}
We define the Lagrangian $l:X \times Y^{*} \times C^{+} \setminus \{ 0
\} \to \mathcal{G}(C)$ of the problem \eqref{prob:P_e} by
\begin{align*}
l(x, y^{*}, z^{*}) = f(x) \oplus \bigcup_{y \in g(x)} S_{(y^{*},
z^{*})}(y) = f(x) \oplus \inf \{S_{(y^{*}, z^{*})}(y) | y \in g(x)
\}.
\end{align*}
\end{definition}

We can recover the primal objective from the Lagrangian.
\begin{theorem}{\cite{HaLo14}[Prop 2.1]} \label{theor:PrLag}
If $f(x) \neq Z$ for each $x \in X$, then
$$\sup_{(y^{*}, z^{*}) \in Y^{*} \times C^{+} \setminus \{0 \}}
l(x,y^{*}, z^{*}) = \bigcap_{(y^{*}, z^{*}) \in Y^{*} \times C^{+}
\setminus \{ 0 \}} l(x, y^{*}, z^{*}) = \left\{ \begin{array}{cc} f(x)
: & 0 \in g(x) \\ \emptyset: & \text{ otherwise} \end{array} \right.$$
\end{theorem}

Under the condition in Theorem \ref{theor:PrLag}, we can formulate
the problem \eqref{prob:P_e} as
$$\inf_{x \in X} \sup_{(y^{*}, z^{*}) \in Y^{*} \times C^{+} \setminus
\{ 0 \}} l(x,y^{*}, z^{*}).$$
We define the dual problem
\begin{align*} \label{prob:D_e} \tag{$\mathcal{D}_{e}$}
\sup_{(y^{*}, z^{*}) \in Y^{*} \times C^{+} \setminus
\{ 0 \}} \inf_{x \in X} l(x,y^{*}, z^{*}).
\end{align*}
We denote by $h:Y^{*} \times C^{+} \setminus \{0\} \to G(C)$ the dual
objective
\begin{equation} \label{duallag}
h(y^{*}, z^{*}) := \inf_{x \in X} l(x, y^{*}, z^{*})
\end{equation}
and define $\overline{d}$ to be the set 
$$\overline{d}:= \sup_{(y^{*}, z^{*}) \in Y^{*} \times C^{+} \setminus \{ 0 \}} h(y^{*}, z^{*}).$$

We have the following weak duality results for the problems \eqref{prob:P_e} and
\eqref{prob:D_e}.
\begin{proposition}{\cite[Prop 6.2]{HaLo14}} \label{prop:WeakDual} 
Weak duality always holds for the problems \eqref{prob:P_e} and
\eqref{prob:D_e}. That is,
$$\overline{d} = \sup \left \{ h(y^{*}, z^{*}) | y^{*} \in Y^{*}, z^{*} \in
C^{+} \setminus \{0 \} \right \} \supseteq \inf \left \{ f(x) | x \in
X, 0 \in g(x) \right \} = \overline{p}.$$
\end{proposition}

Strong duality, on the other hand, requires a constraint
qualification. The problem \eqref{prob:P_e} is said to satisfy the
\emph{Slater condition} if
$$\exists \overline{x} \in \text{ dom } f: \; \; g(\overline{x}) \cap
\setint(-D) \neq \emptyset.$$
Slater's condition is sufficient for strong dualilty between \eqref{prob:P_e} and
\eqref{prob:D_e}.
\begin{theorem}{\cite[Theorem 6.1]{HaLo14}} \label{theor:StrDual}
Assume $p \neq Z$. If $f:X \to \cG(C)$ and $g: X \to
\cG(D)$ are convex and the Slater condition for problem \eqref{prob:P_e}
is satisfied, then strong duality holds for \eqref{prob:P_e}. That is,
$$\overline{p} = \inf \left\{f(x) | 0 \in g(x) \right\} = \sup \left\{ h(y^{*}, z^{*}) | y^{*} \in
Y^{*}, z^{*} \in C^{+} \setminus \{ 0 \} \right\}= \overline{d}.$$
\end{theorem}

Lastly, we introduce the notion of a set-valued Fenchel conjugate.
\begin{definition}\label{def:dualf}
The (negative) Fenchel conjugate of a function $f: X \to
\mathcal{P}(Z)$ is the function $-f^{*}: X^{*} \times (C^{+} \setminus
\{0 \}) \to \mathcal{G}(C)$ defined by
$$- f^{*}(x^{*}, z^{*}) = \cl \bigcup_{x \in X} \left[f(x) + S_{(x^{*},
z^{*})}(-x) \right].$$
\end{definition}
Motivation for this definition and further details about the nature of
the set-valued Fenchel conjugate can be found in \cite{Ham09}.

\section{Duality in Portfolio Optimization}

In this section we apply the tools introduced in the previous section 
for dualizing set-valued optimization problems to the portfolio optimization
problem \eqref{prob:P}.

We begin by showing that \eqref{prob:P} is well-defined.
\begin{lemma}\label{lem:WellDef}
 The functions $F: L^{P} (\cF_{T}, \mathbb{R}^d) \to
\cG(\mathbb{R}^{d}, \mathbb{R}_{+}^{d})$, $F(x) = \mathbb{E}[-U(x)] +
\mathbb{R}^{d}_{+}$ and $g: L^{P} (\cF_{T},
\mathbb{R}^d) \to \cG(L^p(\cF_T,\R^d), L^p(\cF_T,\R^d_+))$, $g(x) = x
- A_{T}(x_0)$ are well-defined and convex.
\end{lemma}
\begin{proof}
We begin with the function $F$. $F$ clearly maps to
$\cG(\mathbb{R}^{d}, \mathbb{R}_{+}^{d})$ because
$$F(x) + \R^{d}_{+} = \mathbb{E}[-U(x)] + \R^{d}_{+}$$
is a polyhedral convex cone, and hence is a closed convex cone.
We claim that $F$ is also a convex map. More precisely, let $x_1, x_2 \in \R^d$, and $t \in (0,1)$. Then
\begin{align*}
& t f(x_1) + (1-t) f(x_2) \\
&= t \left( \mathbb{E}[-U(x_1)] + \R_{d}^{+}
\right) + (1-t) \left( \mathbb{E}[-U(x_2)] + \mathbb{R}_{d}^{+}
\right)\\ 
&= \mathbb{E} \left[ t \left(-U(x_1) \right) + (1-t)
\left(-U(x_2) \right) \right] + \mathbb{R}_{d}^{+}.
\end{align*}
By the assumptions on our objective function $U$, for each $1 \leq i
\leq d$, $-u_i$ is convex, so that
$$t \left(-u_i(x_1(\omega)) \right) +
(1-t)\left(-u_i(x_2(\omega))\right) \geq -u_i \left(t
x_1(\omega) + (1-t)x_2(\omega) \right)$$
for each $\omega \in \Omega$. It follows that
$$\mathbb{E} \left[ t \left(-U(x_1) \right) + (1-t)
\left(-U(x_2) \right) \right] + \mathbb{R}_{d}^{+} \subseteq
\mathbb{E}\left[-U(t x_1 + (1-t) x_2) \right] + \mathbb{R}_{d}^{+}.$$
We conclude that $F$ is convex by Definition \ref{def:Convex}.

Next we consider the function $g(x)$. 
We need to show that in
$L^p(\cF_T,\R^d)$, 
\begin{equation*}
\cl\conv(x-A_T(x_0)+L^p(\cF_T,\R^d_+))=x-A_T(x_0)\,.  
\end{equation*}
Observe that $-A_T(x_0)$ and $L^p(\cF_T,\R^d_+)$ are convex, so their
sum is as well \cite[Ch. 3]{ConvAnal} and the convex
hull on the left side can be dropped. In addition, since $K_T$ is a
solvency cone, the cone $L^p(\cF_T,\R^d_+)$ is contained in
$L^p(\cF_T,K_T)$, thus $x-A_T(x_0)+L^p(\cF_T,\R^d_+)=x-A_T(x_0)$.
Hence,
it remains to show that $x-A_T(x_0)$ is closed in $L^p(\cF_T,\R^d)$.
By the assumptions \eqref{assumptions}, $\Omega$ is finite, and each 
element in $\Omega= \omega_1,\dots, \omega_N$ has positive
probability. The space of $d$-dimensional random variables can then be
associated with Euclidean space of dimension $d \times N$ and inner
product
$\mathbb{E} \langle x, y \rangle$. Note that if $G =
\textit{cone}(\xi_1,...,\xi_m)$
is a random convex cone generated by $m$ $\mathcal{F}_t$-measurable random
variables, then $L^p(G, \mathcal{F}_t)$ is the polyhedral convex cone
generated by
$\xi_i I_{\Gamma_{j}}$ where $\{\Gamma_j\}_{j \in J}$ are the atoms of
$\mathcal{F}_t$.
We have established that each of the $L^p (\mathcal{F}_t, K_t)$ is
finitely
generated, so by the Farkas-Minkowski-Weyl Theorem, each is
polyhedral.
Since the finite sum of polyhedral cones is a polyhedral cone, we
conclude
that $x-A_T(x_0)$ is a polyhedral cone, and hence is closed.
\end{proof}

Note that, in the notation of the previous section, we have
established that $X= L^p(\cF_{T}, \R^d)$, $Y= L^p(\cF_{T}, \R^d)$,
$Z=\R^d$, $C=\R^{d}_{+}$, and $D=L^p(\cF_T, \R^d_+)$.

We are now ready to state our main result, which is a formulation
of the dual problem to the portfolio optimization problem \eqref{prob:P}. 
We then examine the relationship between the primal and dual problems.
Namely, we establish that strong duality holds.

\begin{theorem} \label{thm:DualProb}
The dual problem to \eqref{prob:P} is the problem
\begin{align*} 
\label{prob:D} \tag{$\mathcal{D}$}
\sup_{(y^{*}, z^{*}) \in -A_T(x_0)^{+} \times \R_{+}^{d} \setminus
\{0\}} h(y^{*}, z^{*})
\end{align*}
where $h: L^{q}(\cF_{T}, \R^{d}) \times \R^{d}_{+} \setminus \{0 \}
\to \cG(\R^{d}_{+})$ is defined as
\begin{equation} \label{dualobj}
	h(y^*,z^*)= \left\{ 
		\begin{array}{ll} 
			\{z\in\R^d | \inf_{x \in L^{p}(\cF_T, \R^d)} z^{*}(\E[-U(x)]) +y^{*}(x) \leq z^{*}(z) \} & \text{ if } y^* \in(-A_T(x_0))^+ 
			\\ \R^d&\text{ otherwise}.  
		\end{array} 
		\right.
\end{equation}
When $y^{*} \in -A_{T}(x_{0})^{+}$ and no components of $z^{*}$ are zero, 
we can write the function $h(y^{*}, z^{*})$ as
\begin{equation}
\left\{ z \in \R^d | \mathbb{E} \left[ \sum_{i=1}^{d} z_{i}^{*} \mathbbm{1} 
u_{i}^{*}\left(\frac{y_{i}^{*}}{z_{i}^{*} \mathbbm{1}} \right) \right] \leq z^{*} z \right\}
\end{equation}
where $u_{i}^{*}$ denotes the concave conjugate of $u_i$ \cite{ConvAnal}.
\end{theorem}

To prove this result, we require the following lemma
\begin{lemma} \label{lemma:Lag}
The lagrangian function for the problem \eqref{prob:P} is
$l:L^p(\cF_T,\R^d)\times L^q(\cF_T,\R^d)\times
(\R^d_+\setminus\{0\})\to\cG(\R^d,\R^d_+)$ defined by
\begin{equation}
\label{eqn:lag} 
	l(x, y^{*}, z^{*}) = \left\{ 
	\begin{array}{ll} 
		\{z\in\R^d \,|\,z^{*}(\E[-U(x)]) +y^{*}(x) \leq z^{*}(z) \} & \text{ if } y^{*} \in -A_{T}(x_0)^{+} 
		\\ \R^d & \text{ otherwise}.  
	\end{array} \right.  
\end{equation}
\end{lemma}
\begin{proof}

Note that the positive dual cones of $L^p(\cF_T,\R^d_+)$
and $\R^d_+$ are $L^q(\cF_T,\R^d_+)$ and $\R^d_+$ respectively. Hence
the
Lagrangian function $l$ has domain on $L^p(\cF_T, \R^d) \times
L^q(\cF_T,\R^d)\times
(\R^d_+\setminus\{0\})$, as per Definition \ref{def:lag}.

Also from Definition \ref{def:lag}, we see that \begin{align}\label{eqn:lagtmp}
l(x,
y^{*}, z^{*})= F(x) \oplus \bigcup_{y \in x-A_{T}(x_{0})} S_{y^{*},
z^{*}}(y) \end{align} The union on the right side can be written
explicitly
\begin{align*} \bigcup_{y\in x-A_T(x_0)}\{z\in\R^d\,|\,y^*(y)\le
z^*(z)\}
&=\{z\in\R^d\,|\,\inf_{y\in x-A_T(x_0)}y^*(y)\le z^*(z)\}
\\&=\{z\in\R^d\,|\,\inf_{y\in  -A_T(x_0)}y^*(y)+y^*(x) \le
z^*(z)\}\,.
\end{align*} Note that $A_T(x_0)= x_0-K_0-L^p(\cF_1,K_1)-\cdots-L^p(\cF_T,K_T)$
is a
cone in $L^p(\cF_T,K_T)$. The infimum on the right side
is the support function on a cone \cite{ConvAnal}, and can be
written
\begin{equation*} 
\inf_{y\in -A_T(x_0)}y^*(y)=-\delta_{-A_T(x_0)^+}(y^*)\,, 
\end{equation*}
where $\delta_{-A_{T}(x_0)^{+}}(y^{*})$ is the indicator function on
$-A_{T}(x_0)^{+}$,
equal 0 if $y^{*}$ belongs to $-A_{T}^{+}$ and  $\infty$
otherwise. Hence, identity \eqref{eqn:lagtmp} becomes
\begin{align*}
l(x,y^*,z^*)&=\E[-U(x)]+\{z\in\R^d\,|\,-\delta_{-A_T(x_0)^+}(y^*)+y^*(x) 
\le z^*(z)\}\oplus \R^d_+ \\
&=\{z\in\R^d\,|\,-\delta_{-A_T(x_0)^+}(y^*)+z^*(\E[-U(x)])
+y^*(x)\le z^*(z)\}\oplus \R^d_+ 
\end{align*}
Since $z^*\in\R^d_+ \setminus \{0\}$,
\begin{equation*}
\{z\in\R^d\,|\, a\le z^*(z)\}\oplus\R^d_+=\{z\in\R^d\,|\, a\le
z^*(z)\}
\end{equation*}
for any constant $a$ in $\R\cup \{-\infty\}$. It follows that
\begin{equation*} l(x,y^*,z^*) =\{z\in\R^d\,|\,
-\delta_{(-A_T(x_0))^+}(y^*)+z^*(\E[-U(x)]) +y^*(x)\le
z^*(z)\}
\end{equation*}
which deduces identity \eqref{eqn:lag}.
\end{proof}

Recall that the coordinate functions $u_i(x)$ of the utility 
function $U(x)$ are real-valued for each $x \in \mathbb{R}^d$.
Combining this with the fact that $\Omega$ is finite (and hence
the expectation in the problem formulation is finite) yields
the following proposition.

\begin{proposition}
The objective function of \eqref{prob:P} can be recovered from the
lagrangian \eqref{eqn:lag}. That is,
$$ \sup_{(y^{*}, z^{*}) \in L^{q}(\cF_{T}, \R^d) \times \R^{d}_{+}
\setminus \{ 0 \} } l(x, y^{*}, z^{*}) = \left\{ \begin{array}{cc}
\mathbb{E}[-U(x)] + \mathbb{R}^{d}_{+}: & 0 \in x - A_{T}(x_0) \\
\emptyset: & \text{otherwise} \end{array} \right. $$
\end{proposition}
\begin{proof}
This follows immediately from the above comments and an application of Theorem
\ref{theor:PrLag}.
\end{proof}

Using Lemma \ref{lemma:Lag}, we complete the proof of Theorem \ref{thm:DualProb}.
\begin{proof}
From the definition of the dual objective \eqref{duallag}
$$h(y^{*}, z^{*}) = \inf_{x \in L^{p}(\cF_{T}, \R^d)} l(x, y^*,
z^*).$$
If $y^{*} \not \in -A_{T}(x_0)^{+}$, this is $\R^d$. In the case that $y^{*}
\in -A_{T}(x_0)^+$
\begin{align*}
&\inf_{x \in L^{p}(\cF_{T}, \R^d)} l(x, y^*, z^*) \\
&= \inf_{x \in L^{p}(\cF_{T}, \R^d)} \{z\in\R^d\,|\,
-\delta_{(-A_T(x_0))^+}(y^*)+z^*(\E[-U(x)]) +y^*(x)\le
z^*(z)
\} \\
&= \cl \bigcup_{x \in L^{p}(\cF_{T}, \R^d)} \{z| y^{*}(x) + z^{*}(\E[-U(x)]) \leq z^{*}(z) \} \\
&= \cl \{z |  \inf_{x \in
L^{p}(\cF_{T},\R^d)} z^{*}(\E[-U(x)]) + y^{*}(x) \leq z^{*}(z) \}.
\end{align*}
The infimum in the expression above is the Fenchel conjugate of a
sum of convex functions. Since the Fenchel conjugate of a proper convex
function is proper and lower semicontinuous \cite[Theorem
11.1]{VarAnal}, this infimum is attained,
so we can drop the closure from the expression. Thus,
$$=\{z| \inf_{x \in L^{p}(\cF_{T}, \R^d)}
z^{*}(\E[-U(x)]) + y^{*}(x) \leq z^{*}(z)\}$$
and the result is proven.

The final part of the theorem, in which we reformulate the dual objective in terms of 
concave conjugates, follows because
\begin{align*}
\inf_{x \in L^{p}(\cF_{T}, \R^d)} z^{*}(\E[-U(x)]) + y^{*}(x) \\
=\inf_{x \in L^{p}(\cF_{T}, \R^d)} \E[\langle z^{*}\mathbbm{1},- U(x) \rangle] + y^{*}(x) \\
=\inf_{x \in L^{p}(\cF_{T}, \R^d)} \E[\langle z^{*}\mathbbm{1}, -U(x) \rangle + \langle y^{*}, x \rangle] \\
=\inf_{x \in L^{p}(\cF_{T}, \R^d)} \E[\sum_{i=1}^{d} -z_{i}^{*} u_i(x_i) +  y_i^{*} x_i] \\
=\inf_{x \in L^{p}(\cF_{T}, \R^d)} \sum_{\omega \in \Omega} \mathbb{P}[\omega] \left(\sum_{i=1}^{d} -z_{i}^{*} u_i(x_i(\omega))
 + y_i^{*}(\omega) x_i(\omega) \right) \\
\end{align*}
Exploiting separability over the sum gives
$$= \sum_{\omega \in \Omega} \mathbb{P}[\omega] \left(\sum_{i=1}^{d} \inf_{x_{i}(\omega) \in \R} -z_{i}^{*} u_i(x_i(\omega))
 + y_i^{*}(\omega) x_i(\omega) \right)$$
$$= \mathbb{E}\left[\sum_{i=1}^{d} \inf_{x_{i} \in L^{p}(\mathcal{F}, \R)} y_{i}^{*}(x_{i}) - z_{i}^{*} \mathbbm{1} u_{i}(x_{i})\right]$$
from which the result follows immediately. More details can be found in the next section, where we perform
the details of this calculation more slowly with an example problem for context.

\end{proof}

In the language of set-valued Fenchel conjugates, we have the following easy
corollary.
\begin{corollary}
The objective function of the dual problem is
\begin{equation}
h(y^*,z^*)= \left\{ \begin{array}{ll} -F^*(-y^*,z^*) & \text{ if
}
y^*\in -A_T(x_0)^+ \\ \R^d&\text{ otherwise}.  \end{array} \right.
\end{equation}
where $F^*$ is the Fenchel conjugate of $F$.
\end{corollary}

\begin{proof}
We compute the Fenchel conjugate of $F$. For
every $y^*\in L^q(\cF_T,\R^d)$, $z^*\in \R^d_+\setminus\{0\}$, it
follows from Definition \ref{def:dualf} that 
\begin{align} \label{tmp1} 
-F^*(-y^*,z^*) &= \cl \bigcup_{x \in X} \left( F(x) +
S_{(-y^{*}, z^{*})}(-x) \right) \nonumber\\
&=\cl\bigcup_{x\in L^p(\cF_T,\R^d)} \left( \E[-U(x)]+\R^d_+ +\{z|
-y^{*}(-x) \leq z^{*}(z) \} \right) \nonumber\\
&= \cl \bigcup_{x \in L^p( \cF_T, \R^d)} \left\{ z + \mathbb{E}[-U(x)] +
\R^d_{+} | - y^{*} (-x) \leq z^{*}(z) \right\} \\
\end{align}
Since $z^{*} \in \R^{d}_{+} \setminus \{0\}$, $z^{*}(r) \geq 0$ for
each $r \in \R^{d}_{+}$, \eqref{tmp1} becomes
\begin{align*}
&=\cl \bigcup_{x \in L^p( \cF_T, \R^d)} \left\{ z | y^{*}(x) \leq z^{*} (z -
\mathbb{E}[-U(x)] \right\} \\
&=\cl \left\{ z| \inf_{x \in L^p( \cF_T, \R^d)} y^{*}(x) +
z^{*}(\mathbb{E}[-U(x)]) \leq z^{*}(z) \right\} \\
&= \left\{ z| \inf_{x \in L^p( \cF_T, \R^d)} y^{*}(x) +
z^{*}(\mathbb{E}[-U(x)]) \leq z^{*}(z) \right\}.
\end{align*}
This agrees with \eqref{dualobj}.
\end{proof}


\begin{theorem}\label{theor:StrongDual}
Strong duality holds between the problems \eqref{prob:P} and
\eqref{prob:D}. That is, 
\begin{align*}
\overline{p} = &\inf F(x) = &\sup(y^{*}, z^{*}) = \overline{d}\\
&\text{subject to } x \in A_{T}(x_0) &\text{subject to } y^{*} \in
-A_{T}(x_0)^{+} \\ & &z^{*} \in \R^d_+ \setminus \{ 0 \}.
\end{align*}
\end{theorem}
\begin{proof}
By Theorem \ref{theor:StrDual} and Lemma \ref{lem:WellDef}, it
suffices to show that $\overline{p}= \inf_{x \in A_{T}(x_0)} F(x) \neq \R^d$ and that
Slater's condition is satisfied. 

For the first part, we use weak duality. By Proposition
\ref{prop:WeakDual}, $ \overline{p} \subseteq \overline{d}$, so it suffices to show that
$\overline{d} \neq \R^d$. Lemma \ref{lem:KhoaLemma}
give that $-A_{T}(x_0)^{+} \cap \setint(L^{P}(\cF_{T}, \R^{d}_{+})$ 
is nonempty, so there exists $\tilde{y}^{*} \in -A_{T}(x_0)^{+}$ with
$\tilde{y}^{*}(\omega)_{i} <0$ for each $\omega \in \Omega$, $1 \leq i \leq
d$. Then
$$\overline{p} \subseteq \sup_{(y^{*}, z^{*}) \in -A_{T}(x_0)^{+} \times z^{*}
\in \R^{d}_{+} \setminus \{ 0 \}} h(y^{*}, z^{*}) \subseteq \{z | \inf_{x \in L^{p}
(\cF_{T}, \R^d)} \mathbbm{1}_{d} (\mathbb{E}[-U(x)]) + y^{*}_{0}(x) \leq
\mathbbm{1}_{d}(z) \}.$$
The last containment follows by taking $y^{*} = \tilde{y}^{*}$ and $z^{*}
= \mathbbm{1}_{d}$, the $d$-dimensional vector consisting of all ones.
So it suffices to show that 
$$\inf_{x \in L^{P}(\cF_{T}, \R^d)} \mathbbm{1}_{d}(\mathbb{E}[-U(x)]) -
\tilde{y}^{*}(x) > -\infty,$$
which is equivalent to
$$\sup_{x \in L^{P}(\cF_{T}, \R^d)}  
\tilde{y}^{*}(x) -\mathbbm{1}_{d}(\mathbb{E}[-U(x)])< \infty.$$
Note that the left hand side of this expression is simply the
Fenchel-Conjugate of the function $\mathbbm{1}_{d}(\mathbb{E}[-U(x)])$. We
have
\begin{align} \label{tmpexp1}
& \sup_{x \in L^{p}(\cF_{T}, \R^d)}  \tilde{y}^{*}(x)
-\mathbbm{1}_{d}(\mathbb{E}[-U(x)]) \nonumber \\
&= \sup_{x \in L^{p}(\cF_{T}, \R^d)} \mathbb{E}[\langle \tilde{y}^{*}(\omega),
x(\omega) \rangle] - \mathbb{E}[\langle \mathbbm{1}, -U(x(\omega))
\rangle ] \nonumber \\
&= \sup_{x \in L^{p}(\cF_{T}, \R^d)} \sum_{\omega =
\omega_{1}}^{\omega_{n}} \mathbb{P}[\omega] \left[ \sum_{i=1}^{d}
\tilde{y}_{i}^{*}(\omega) x_{i}(\omega) - (-u_{i} (x_{i}(\omega))) \right]
\nonumber \\
&= \sum_{i=1}^{d} \sum_{\omega=\omega_{1}}^{\omega_{n}}
\mathbb{P}[\omega] \left[ \sup_{x_{i}(\omega)} \tilde{y}^{*}_{i}(\omega) -
(-u_{i}(x_{i}(\omega))) \right] .
\end{align}
The first equality follows from the definition of the inner product in
$L^{p}(\cF_{T}, \R^d)$. The second and third come from the finiteness
of $\Omega$ and the separability of the expression, respectively.

Since each $u_i'$ is continuous with range $(- \infty,
0)$, and each $\tilde{y}_{i}^{*}(\omega) <0$, the intermediate
value theorem gives that, for each $\omega \in \Omega$, $1 \leq i \leq
d$, there exists $\tilde{x}_{i}(\omega)$ such that $u_{i}'(\tilde{x}_{i}(\omega)) =
\tilde{y}^{*}_{i}(\omega)$. By \cite[Theorem 23.5]{ConvAnal}, 
$$\sup_{x(\omega)_{i}} \tilde{y}^{*}_{i}(\omega)
-(-u_{i}(\tilde{x}_{i}(\omega)))$$ 
achieves its supremal value at $\tilde{x}_{i}(\omega).$ It follows that
\eqref{tmpexp1} is finite, from which we conclude that $\overline{p}
\neq \R^d$.

Next we show that Slater's condition is satisfied. We want to find an
$\overline{x} \in \dom F$ such that $x-A_{T}(x_0) \cap
\setint(-L^{p}(\F_{T}, \R^{d}_{+})) \neq
\emptyset$. Recall from the problem formulation that $\dom F =
L^{p}(\F_{T}, \R^d)$, so the first part of Slater's condition is not a
restriction. Note that since
$$A_{T}(x_0) = x_0 \mathbbm{1} - K_0 \mathbbm{1} - L^{p}(\cF_1, K_1)
- ... - L^p(\cF_T, K_T)$$
where each $K_t$ is a solvency cone, $x_0 \mathbbm{1} \in A_T(x_0)$.
Then, choose $\overline{x}$ such that $\overline{x}_{i}(\omega) <
(x_0)_{i}$ for each component $1 \leq i \leq d$ and for all $\omega \in \Omega$.
 We have
that $\overline{x} - x_0 \mathbbm{1} \in x - A_{T}(x_0)$ and also
$\overline{x} - x_0 \mathbbm{1} \in \setint(-L^{p}(\cF_{T},
\R_{+}^{d}))$, so Slater's condition is satisfied.
\end{proof}
\section{An Example}
In this final section, we explore an example  which we hope will 
help to illustrate the theoretical results from previous sections. 

We consider a market with 2 assets and 3 time steps, so that the time
step $t$ ranges from 0 to $3$. The probability space $\Omega =
\bigtimes_{i=1}^{3}  \{-1, 0, 1\}$, with the probability measure
$\mathbb{P}$ defined uniformly on this set. In other words, the
possible outcomes $\omega$ are defined as a tuple $\omega = (\omega_1,
\omega_2, \omega_3)$, $\omega_1, \omega_2, \omega_3 \in \{-1, 0 , 1\}$. 
From the decision maker's
perspective, we have that at each time step $t$ the random variable taking
values $\omega_t$ becomes known. Thus the filtration
$(\mathcal(F_{t})_{t=0}^{3})$ is defined by $\mathcal{F}_{t} =
\sigma(\omega_{i} |_{i \leq t})$, the sigma algebra generated by these
random variables. We also take $\mathcal{F} = \mathcal{F}_{3}= 
\sigma(\omega_{1}, \omega_{2}, \omega_{3})$, the sigma-algebra of full information.

The bid-ask process $(\Pi_{t})_{t=0}^{T}$ is defined as follows:
$$\Pi_{t}(\omega) = \left[ \begin{array}{cc} 1 & 1 \\ 8 \cdot
2^{\sum_{i\leq t} \omega_i} & 1 \end{array} \right].$$
This is obviously $(\mathcal{F}_{t})_{t=0}^{3}$ adapted, and one can
also easily check that the properties of bid-ask matrix are satisfied
for each realization. The solvency cones generated by this process are
$$K_{t}(\omega)=\text{cone}\{(1,-1), (-1, 8 \cdot 2^{\sum_{i \leq t}
\omega_{i}}) \} = \{(x,y) | x+y \geq 0,  8 \cdot 2^{\sum_{i \leq t}
\omega_{i}} x + y \geq 0 \}.$$
Figure \eqref{conefig} illustrates these cones for various times and realizations
of $\omega$. 

\begin{figure}
  \includegraphics[width=\linewidth]{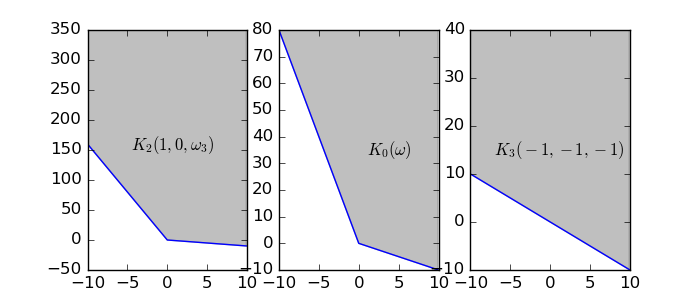}
  \caption{The solvency cones $K_{t}(\omega)$ for various $t$ and
$\omega.$}
  \label{conefig}
\end{figure}

We define our vector-valued objective function to be
$$U(x) = (-e^{-x_{1}}, - e^{-x_{2}})^{T}$$
where $x = (x_{1}, x_{2})^{T}$ is the quantity of physical assets
we have at terminal time.  Our set-valued objective function to be minimized is then
$$F(x) = \mathbb{E}\left[ \left( e^{-x_1} , e^{-x_2}
\right)^{T} \right] + \mathbb{R}_{+}^{2}.$$
We assume that our initial endowment is $(0,0)^T$. The set of self-financing 
portfolios is
$$A_{3} = -K_{0} \mathbbm{1} - L^{p}(\sigma(\omega_1); K_{1}) -
L^{p}(\sigma(\omega_1, \omega_2), K_{2}) - L^{p}(\sigma(\omega_1,
\omega_2, \omega_3), K_{3})$$
where $K_t(\omega)$ are given as above.

We can then formulate the primal portfolio optimization problem as
\begin{align*} \label{prob:P_ex} \tag{$\mathcal{P}_{ex}$}
\text{minimize } \mathbb{E}\left[ \left( e^{-x_1} , e^{-x_2}
\right)^{T} \right] + \mathbb{R}_{+}^{2} \\
\text{subject to } x \in A_{3}
\end{align*}

According to theorem \eqref{thm:DualProb}, the dual problem is then
\begin{align*} \label{prob:D_ex} \tag{$\mathcal{D}_{ex}$}
\sup \{z | \inf_{x \in L^{P}(\R^2, \mathcal{F}_{3})} z^{*}(
\mathbb{E}\left[\left( -e^{x_1} , -e^{x_2} \right)^{T}
\right]) + y^{*}(x) \leq z^{*}(z) \} \\
\text{subject to }  y^{*} \in -A_{3}^{+} \\
 z^{*} \in \R_{+}^{2} \setminus \{ 0 \}
\end{align*}   

First, we investigate the nature of the set $-A_{3}^{+}$. From
\cite{ConvAnal}[Cor 16.4.2],
$$-A_{3}^{+} = (K_{0} \mathbbm{1} + L^{p}(\mathcal{F}_{1}, K_{1})
+L^{p}(\mathcal{F}_{2}, K_{2}) + L^{p}(\mathcal{F}_{3}, K_{3}))^{+}$$
$$= \bigcap_{i=0}^{3} (L^{p}(\mathcal{F}_{i},
K_{i}))^{+}.$$
In other words, $y^{*} \in L^{q}(\mathcal{F}, \R^2)$ is in
$-A_{3}^{+}$ when $y(\omega) \in K_{t}(\omega)^{+}$ for each $t=0,...,3$. 

We can explicitly compute the cones $K_{t}(\omega)^{+}$. We have that 
$$K_{t}(\omega) = \text{cone}\left(\text{conv}\left\{ \left(\begin{array}{c} 1 \\ -1
\end{array}\right), \left(\begin{array}{c} -1 \\ 8 \cdot
2^{\sum_{i \leq t} \omega_{i}} \end{array} \right)  \right\} \right)^{+}.$$
Hence the dual cone is
$$K_{t}(\omega)^{+} = \text{cone}\left(\text{conv}\left\{
\left(\begin{array}{c} -1 \\ -1 \end{array}\right),
\left(\begin{array}{c} - 8 \cdot 2^{\sum_{i \leq t} \omega_{i}} \\ -1 \end{array} \right)  
\right\} \right).$$
Next we take the intersection of these cones to form
$A_{3}^{+}(\omega)$. For a fixed $\omega$, let $s(\omega) =
\min_{j=0,1,2,3} \sum_{i=1}^{j} \omega_{i}$. Then
$$-A_{3}(\omega)^{+} = \text{cone}\left(
\left(\begin{array}{c} -1 \\ -1 \end{array}\right),
\left(\begin{array}{c} - 8 \cdot 2^{s(\omega)} \\ -1 \end{array} \right)  
 \right).$$
Figure \eqref{DualConeFig} illustrates $K^{+}_{t}(\omega)$ for various
times and realizations of $\omega$.

\begin{figure}
  \includegraphics[width=\linewidth]{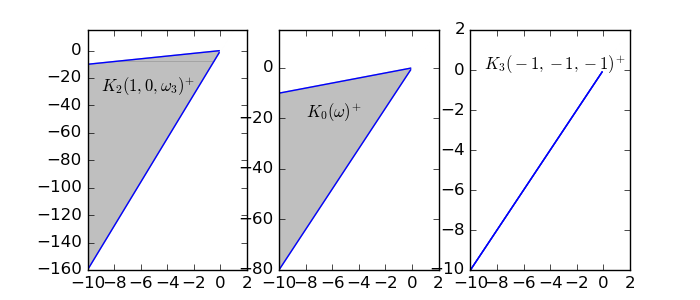}
  \caption{The positive polar cones $K_{t}(\omega)^{+}$ for various $t$ and
$\omega.$}
  \label{DualConeFig}
\end{figure}

Now we work with to simplifiy the dual problem \eqref{prob:D_ex}. The
objective function is
\begin{align*}
h(y^{*}, z^{*}) = 
\left\{ z| \inf_{x \in L^{p}(\mathcal{F}, R^2)}  z^{*}(\mathbb{E}\left[(e^{-x_1},
e^{-x_2})^{T} \right]) + y^{*}(x) \leq z^{*}(z) \right\} \\
= \left\{z| \inf_{x \in L^{p}(\mathcal{F}, R^2)}  \mathbb{E}\left[z_{1}^{*}
\mathbbm{1} e^{-x_1}+ z_{2}^{*} \mathbbm{1} e^{-x_2} \right] +
y^{*}(x) \leq z^{*}(z) \right\}
\end{align*}
Recall that the nature of linear functionals 
$y^{*} \in L^{q}(\mathcal{F}, \mathbb{R}^d)$ is $y^{*}(x) =
\mathbb{E}\langle y^{*}, x \rangle$. Hence the objective becomes
\begin{align*}
\left\{z| \inf_{x \in L^{p}(\mathcal{F}, R^2)}  \mathbb{E}\left[z_{1}^{*}
\mathbbm{1} e^{-x_1}+ z_{2}^{*} \mathbbm{1} e^{-x_2}  +
y_{1}^{*}x_{1} + y_{2}^{*} x_{2} \right] \leq z^{*}(z) \right\}
\end{align*}
Using the fact that our probability space is finite, we expand the
expectation
\begin{align*}
\left\{z| \inf_{x \in L^{p}(\mathcal{F}, \R^2)}  \frac{1}{27}
\sum_{\omega \in \Omega} z_{1}^{*} e^{-x_1(\omega)}+ z_{2}^{*}
 e^{-x_2(\omega)}  +
y_{1}^{*}(\omega) x_{1}(\omega) + y_{2}^{*}(\omega) x_{2}(\omega) 
\leq z^{*}(z) \right\}
\end{align*}

This infimum is separable over the $x_{i}(\omega)$ variables. Hence,
\begin{align*}
=\left\{z| \frac{1}{27} \sum_{\omega \in \Omega}
 \inf_{x_{1}(\omega) \in \R} \{ z_{1}^{*}
 e^{-x_1(\omega)} +y_{1}^{*}(\omega) x_{1}(\omega)\} +
\inf_{x_{2}(\omega) \in \R} z_{2}^{*}  e^{-x_2(\omega)}  +
  y_{2}^{*}(\omega) x_{2}(\omega)  \leq z^{*}(z) \right\}
\end{align*}
We can compute each of these infimums explicitly. Note that 
$$\inf_{x_1(\omega) \in \R} z_{1}^{*} e^{-x_1(\omega)} +
y_{1}^{*}(\omega)
x_{1}(\omega) = -f^{*}(-y_{1}^{*}(\omega))$$
where $f(x) = z_{1}^{*} e^{-x}$ and $f^{*}$ denotes the convex
conjugate of $f$. Recall that for $g :
\mathbb{R} \to \mathbb{R}$ and $a \in \R$ (See \cite{ConvAnal})
\begin{align*}
&(a g(\cdot))^{*}(x^{*}) = a g^{*}(x^{*}/a) \\
&(g(a \cdot ))^{*}(x^{*}) = g^{*}(x^{*}/a) \\
g^{*}(x^{*}) &= \left\{ \begin{array}{ccc} x^{*} \ln(x^{*}) - x^{*}
& \text{ if } x^{*}>0 \\ 0 & \text{ if } x^{*} = 0 \\ \infty \text{ 
otherwise}\end{array} \right.
\text{ when } g(x) = e^{x}.
\end{align*}
Hence the objective function becomes
\begin{align*}
\left\{ z | \frac{1}{27} \sum_{\omega \in \Omega} y _{1}^{*}(\omega)
 - y_{1}^{*}(\omega) \ln \left(\frac{y_{1}(\omega)^{*}}{z_{1}^{*}} \right) 
+ y_{2}^{*}(\omega) - y_{2}^{*}(\omega) \ln
\left( \frac{y_{2}^{*}(\omega)}{z_{2}^{*}} \right)  \leq z^{*} (z) \right\}
\end{align*}
when $y_1(\omega), y_2(\omega), z_1, z_2 \neq 0$. 

Therefore, when $y_1(\omega), y_2 (\omega), z_1, z_2 \neq 0$, we have
the following formulation of the dual problem: find the supremum of
\begin{align*} 
\left\{ (z_1,z_2)^{T} \left| \frac{1}{27} \sum_{\omega \in \Omega} y _{1}^{*}(\omega)
 - y_{1}^{*}(\omega) \ln \left(\frac{y_{1}^{*}(\omega)}{z_{1}^{*}} \right) 
+ y_{2}^{*}(\omega) - y_{2}^{*}(\omega) \ln
\left( \frac{y_{2}^{*}(\omega)}{z_{2}^{*}} \right)  \leq z_1^{*} z_1 +
z_{2}^{*} z_2 \right\} \right.
\end{align*}
\begin{align*}
\text{subject to } & \left( \begin{array}{c} z^{*}_1 \\ z^{*}_2
\end{array} \right) 
\in \R^2 \setminus \left\{ \left( \begin{array}{c} 0 \\ 0 \end{array}
\right) \right\}
\\ 
& \left( \begin{array}{c} y^{*}_{1}(\omega) \\ y^{*}_{2}(\omega) \end{array}
\right) \in A_{3}(\omega) \; \; \forall \omega \in \Omega 
\end{align*}

When either $z_1$ or $z_2$ equals $0$, we only consider $x_2$ or $x_1$, respectively, in 
our objective because the other terms vanish. Likewise, the case that $y_1(\omega) = 0$ eliminates
the expressions with $y_1(\omega)$ in the objective, because of the conjugation result on 
the previous page. The case that $y_{2}(\omega)=0$ is symmetric. 

\printbibliography

\end{document}